\documentclass[preprint,11pt]{article}

\usepackage{amssymb,amsfonts,amsmath,amsthm,amscd,dsfont,mathrsfs}
\usepackage{graphicx,float,psfrag,epsfig}
\usepackage{wrapfig}

\DeclareMathAlphabet{\mathpzc}{OT1}{pzc}{m}{it}

\newtheorem{propo}{Proposition}[section]
\newtheorem{lemma}[propo]{Lemma}

\newtheorem{thm}[propo]{Theorem}

\newtheorem{remark}[propo]{Remark}

\def\endproof{\hfill$\Box$\vspace{0.4cm}}

\def\di{{\partial i}}

\def\tG{\widetilde{\Gamma}}

\newcommand{\sign}{\text{sign}}
\newcommand{\reals}{{\mathds R}}

\newcommand{\eqnsection}{\renewcommand{\theequation}{\thesection.\arabic{equation}}
      \makeatletter \csname @addtoreset\endcsname{equation}{section}\makeatother}
\def\eps{\epsilon}
\def\cC{{\cal C}}
\def\l|{\left|\left|}
\def\r|{\right|\right|}

\def\E{\mathds E}
\def\1{\mathds 1}
\def\prob{{\mathds P}}

\def\ind{{\mathds I}}

\def\ve{\varepsilon}

\def\de{{\rm d}}
\def\reals{{\mathds R}}

\def\ux{\underline{x}}
\def\uy{\underline{y}}
\def\uz{\underline{z}}

\def\Eu{\mbox{\tiny\rm Eucl}}

\def\cut{{\rm cut}}
\def\uh{\underline{h}}

\def\up{\underline{+1}}

\def\Var{{\rm Var}}
\def\Dir{{\rm Dir}}

\def\dOmega{\partial\Omega}
\def\cS{{\cal S}}

\DeclareMathOperator*{\argmin}{arg\,min}

\begin{document}
\begin{titlepage}
\title{Convergence to Equilibrium in Local Interaction Games \\
and Ising Models}

\author{Andrea Montanari\thanks{~Departments of  Electrical Engineering and Statistics, Stanford University
\newline\indent
${}^\dagger$Department of Management Science and Engineering,
Stanford University} \and  Amin Saberi$^\dagger$}

\date{\today}

\maketitle

\begin{abstract}
Coordination games describe social or economic interactions in which
the adoption of a common strategy has a higher payoff. They are
classically used to model the spread of conventions, behaviors, and
technologies in societies. Here we consider a two-strategies
coordination game played asynchronously between the nodes of a
network. Agents behave according to a noisy best-response dynamics.

It is known that noise removes the degeneracy among equilibria: In
the long run, the ``risk-dominant'' behavior spreads throughout the
network. Here we consider the problem of computing the typical time
scale for the spread of this behavior. In particular, we study its
dependence on the network structure and derive a dichotomy between
highly-connected, non-local graphs that show slow convergence, and
poorly connected, low dimensional graphs that show fast convergence.
\end{abstract}
\end{titlepage}

\pagestyle{empty}

\section{Introduction}
The unprecedented growth of online social network and their
increasing role in the spread of knowledge, behaviors and new
technologies have given rise to a wealth of interesting questions.
Is it possible to explain the emergence of a new phenomenon based on
the dynamics of the interaction among individuals
\cite{KleinbergReview,Young}?


%
\begin{wrapfigure}[11]{r}{5cm}
\begin{picture}(100,100)
\put(0,-5){\includegraphics[scale=0.33]{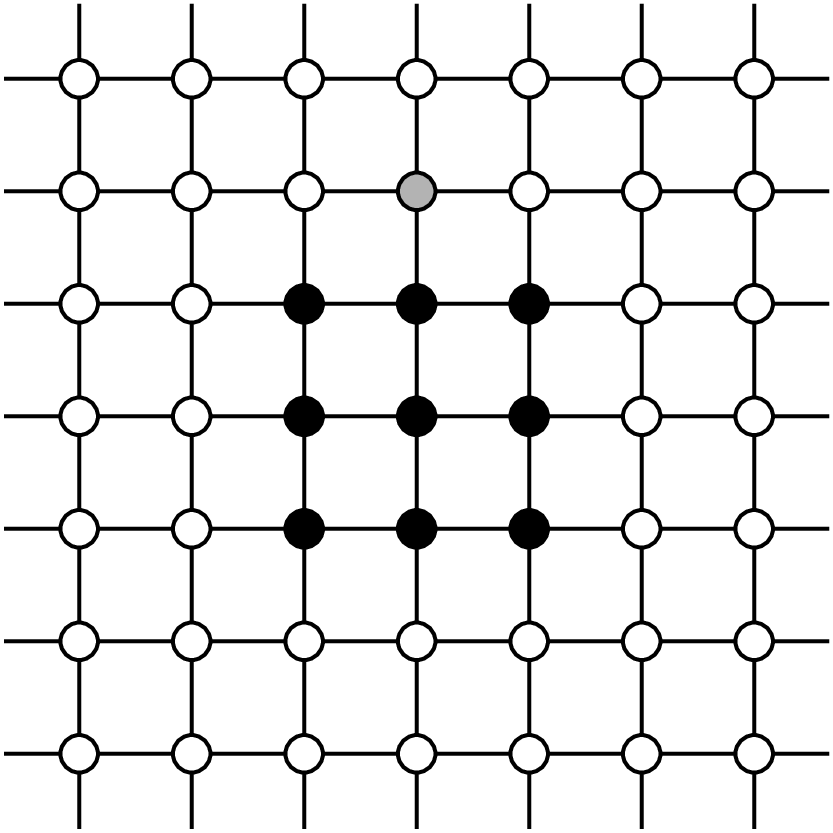}}
\end{picture}
\end{wrapfigure}

As an example consider a two-dimensional grid and assume that each
node adopts  the new behavior (call it $+1$, the alternative being
$-1$) if at least two of its neighbors have already adopted it. It
is then easy to see that no finite set of $+1$'s can influence the
whole grid, and in fact the influence of any finite set of $+1$'s is
limited to the smallest rectangle that circumscribes them. For
instance, the group of black nodes in the figure on the right does
not expand further.

Now, consider the same dynamics with a small noise, i.e. assume
that, with some small probability $\eps$, agents do not follow the
pre-established rule. This can have dramatic effects. If the gray
node in the figure switches to $+1$ by mistake, then a new layer may
be added to the group of black nodes at no extra (probability) cost.
Of course, the reverse can happen: the block of $+1$'s can be eroded
because of noise. However if the initial block is \emph{large
enough}  (and under some technical assumptions) the former mechanism
will prevail \cite{Neves1,Neves2}. The important point is that
`large enough' means here larger than some \emph{constant} quantity,
and that influence spreads at some positive velocity. This
phenomenon was first discovered in statistical physics, under the
name of `nucleation' and received an intense attention in the
mathematical physics literature over the last 30 years
\cite{Olivieri,Bovier}.

Similar models were developed independently within the context of
evolutionary game theory. For example consider a simple game in
which every individual placed in a  network has to make a decision
between two alternatives. The payoff of an action for each person is
proportional to the number of its neighbors who are taking the same
action. These games, known as coordination games, have been studied
extensively for modeling the emergence of technologies and social
norms \cite{Young,Morris,KleinbergReview,Blume}. The main conclusion
of this line of work is that adding a small random perturbation to
best response dynamics creates an evolutionary force that drives the
system towards a particular equilibrium in which all players take
the same action.

In real-world networks stochasticity  is unavoidable. As a
consequence, we can expect the players to eventually achieve
coordination on a particular equilibrium, irrespective of the initial
state. The present paper
characterizes the {\em rate of convergence} for such dynamics   in
terms of explicit graph quantities. It thus provide the first step
in a longer term program aimed at developing approximation
algorithms to estimate convergence to Nash equilibria.

Our characterization is expressed in terms of tilted cutwidth and
tilted cut of the graph that are dual quantities. The former
provides a path to the $+1$ equilibrium that gives an upper bound on
the converge time. The latter corresponds to a bottleneck along the
highest separating set in the space of configurations. We show that
tilted cut and tilted cutwidth coincide for the `slowest' subgraph
and the convergence time is exponential in this graph parameter.

The proof uses an argument similar to
\cite{Donsker,DiaconisSaloffCoste,JerrumSinclair} to relate hitting
time to the spectrum of an appropriate transition kernel. The
convergence time is then estimated in terms of the most likely path
from the worst-case initial configuration. It turns out  that the
most likely path is the one that implies the lowest decrease of
probability in stationary measure. A delicate argument using the
submodularity of the potential function shows that there exists a
monotone increasing path with this property. In order to prove the
characterization in terms of tilted cut we study the `slowest'
eigenvector and show that it is monotone using a fixed point
argument. We then approximate the eigenvector with a characteristic
function.

The above result allows us to estimate the convergence time for
specific graphs through their isoperimetric function. For example in
interaction graphs that can be embedded in low dimensional spaces,
the dynamics converges in a very short time. On the other hand, for
a wide class of bounded degree graphs such as random regular graphs
or certain small-world networks the convergence may take as long as
exponential in the number of nodes. \vspace{0.4cm}

{\noindent \bf Related work} \vspace{0.2cm}

\noindent There is a very interesting line of work in mathematical
physics leading to very sharp estimates of the convergence times of
specific models: mainly two and three dimensional grids
\cite{BenArous,BovierManzo}. Berger et al. \cite{PeresEtAl} compute
the mixing time of a similar dynamics in terms of cutwidth of the
graph using different techniques from the current paper.

In the game theory literature, one of the criticisms of Nash
equilibria is that its multiplicity makes it hard to predict the
outcome of a play. How do players \emph{learn} to play a specific
equilibrium, and which one do they \emph{select}?  For example, the
grid graph described above shows that the coordination game can have
several equilibria. There is a vast literature in evolutionary game
theory for resolving this problem especially in the context of
coordination games \cite{KMR,Young,Ellison,Blume,FLbook}.


The importance of estimating convergence times  was first stressed
in the pioneering work of Ellison \cite{Ellison}. He argued that the
long-run equilibrium is relevant only if the convergence time is
reasonably small. Ellison studied the rate of convergence for two
extreme interaction graphs: a complete graph and a graph obtained by
 placing individuals on a cycle and connecting all pairs of distance
 smaller than some given range. He showed that the dynamics converges very
 slowly for the former model and very quickly for the latter. Based on
this observation, he concluded that when the interaction is  global
the outcome is determined by  historic factors. In contrast, when
players ``interact with small sets of neighbors,'' we can assume
that evolutionary forces may determine the outcome.

Our result implies that the key property of the network that
captures the rate of convergence is not the number of nodes each
agent interacts with, or the number of edges of the graph. This can
be proved for a large class of (non-reversible) noisy best-response
dynamics including the one of \cite{Ellison}.

\section{Definitions}

A game is played in periods $t = 1,
2,3,\ldots$ among a set $V$ of players.
Each player $i \in V$ has
two alternative strategies as $x_i \in \{+1,-1\}$.
Let $\ux = \{x_i:\, i\in V\}$. The payoff matrix $A$ is a
$2\times 2$- matrix illustrated in the figure.
The players interact on an undirected graph $G=(V, E)$.
The payoff of  player $i$ is $\sum_{j \in \di}{A(x_i, x_j)}$,
where $\di$ is the set of neighbors of vertex $i$.

\begin{wrapfigure}[6]{r}{3cm}
\begin{picture}(40,40)
\put(10,0){\includegraphics[scale=0.33]{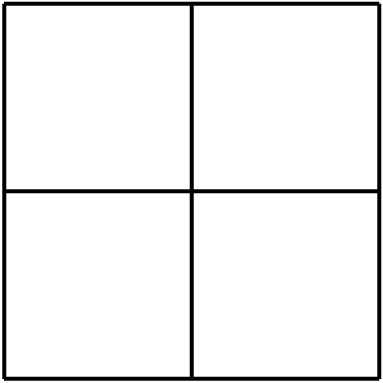}}
\put(17,44){$a, a$}
\put(49,44){$c, d$}
\put(17,12){$d, c$}
\put(49,12){$b, b$}
\end{picture}
\end{wrapfigure}

The payoff matrix $A$ defines a coordination game which means $a >
d$ and $b > c$.  It is easy to verify that for every $i$, the best
response strategy is $\sign(h_i+\sum_{j\in\di}x_j)$, where $h_i=
\frac{a-d-b+c}{a-d+b-c}|\di|\equiv \rho\,|\di|$, with $|\di|$ the
degree of node $i$. We assume that $a - b > d - c$, so that $h_i >
0$ for all $i\in V$ of non-vanishing degree.  Harsanyi and Selten
\cite{HS} named $+$ the ``risk-dominant'' equilibrium, as it
minimizes the utility loss due to a change in the opponent strategy.
Notice that this does not coincide, in general, with the payoff
dominant equilibrium.

Noisy best-response dynamics is specified by a
one-parameter family of Markov chains
$\prob_{\beta}\{\,\cdots\,\}$ indexed by $\beta$.
The parameter $\beta\in\reals_+$ determines how noisy is the dynamics,
with $\beta = +\infty$ corresponding to the noise-free case.
Two type of updates  are naturally defined:
%
%

\vspace{0.1cm}

\noindent\emph{(1) Synchronous updates}. At each step of the chain, each
player draws a new strategy $y_i$ conditionally on its neighbor's
strategies $x_{\di}$ at the previous time step. The conditional distribution
is denoted by $p_{i,\beta}(y_i|\ux_{\di})$.

\vspace{0.1cm}

\noindent \emph{(2) Asynchronous updates}. Each node $i$ updates its value at the arrival time of an independent Poisson clock of rate $1$. The conditional distribution of the new strategy is
again denoted as $p_{i,\beta}(y_i|\ux_{\di})$.

\vspace{0.1cm}

\noindent
The dynamics  of \cite{Ellison} is recovered by the following transition probabilities. Let
$y_i^* = \sign(h_i+\sum_{j\in\di}x_j)$. Then for every player $i$,  $p_{i,\beta}(y_i^*|\ux_{\di})=1-e^{-\beta}$ and
$p_{i,\beta}(-y_i^*|\ux_{\di})=e^{-\beta}$.

A considerable simplification is achieved for the so-called
\emph{heath bath} or \emph{Glauber} kernel
\begin{eqnarray}
p_{i,\beta}(y_i|\ux_{\di}) = \left(1+e^{-2\beta K_i(\ux)y_i}\right)^{-1}
\end{eqnarray}
where $K_i(\ux) = h_i+\sum_{j\in\di}x_j$. . This is also known as
logit update rule which is standard in the discrete choice
literature \cite{McFadden}. It has also been used to model subjects'
empirical choice behavior in laboratory situations \cite{MS,MP}. In
this context it has been studied by Blume \cite{Blume}. The
corresponding Markov chain is reversible with respect to the
stationary distribution $\mu_{\beta}(\ux) \propto \exp(-\beta
H(\ux))$, with

\vspace{-0.9cm}

\phantom{a}

\begin{eqnarray}
H(\ux) & = & -
\sum_{(i,j)\in E} x_ix_j-\sum_{i\in V} h_i x_i\, ,\label{eq:EnergyAsynchron}
\end{eqnarray}

\vspace{-0.5cm}\phantom{a}

\noindent in the case of asynchronous dynamics. This is the energy function of
the Ising model; an analogous expression can be written for synchronous
updates. In both the above models
the stationary distribution for large $\beta$   concentrates
around the all-$(+1)$ configuration. In other words, these dynamics predict
 that
in the long run, the play will converge to the risk-dominant equilibrium.

In the following we will often adopt the equivalent representation
of configurations as subsets of vertices $S\subseteq V$, whereby
$i\in S$ if and only if $x_i=+1$, and, with a slight abuse of
notation, we shall denote by $H(S)$ the corresponding energy. If
$|S|_h \equiv \sum_{i\in S}h_i$, then $H(S) - H(\emptyset)=
2\, \cut(S,V\setminus S)- 2|S|_h$. It is important to notice that
$H(\,\cdot\,)$ is submodular.

Our aim is to determine whether this
prediction is realized in a reasonable time.
To this end, we let $T_+$ denote the hitting time to the all-$(+1)$
configuration, and define the \emph{typical hitting time} for
$\up$  as

\vspace{-1.0cm}

\phantom{a}

\begin{eqnarray}
\tau_+(G;\uh) = \sup_{\ux}\;
\inf\left\{t\ge 0:\, \prob_{\beta}^{\ux}\{T_+\ge t\}\le e^{-1}
\right\} \, .\label{eq:TypicalHitting}
\end{eqnarray}
For the sake of brevity, we will often refer to this as the hitting
time, and drop its arguments.

\section{Main results}

Our first step is to express the large-$\beta$ (low-noise)
behavior of $\tau_+(G;\uh)$ in terms of  graph-theoretical quantities.
Let $n=|V|$ be the number of players.
Given $\uh = \{h_i:\, i\in V\}$, and $U\subseteq V$,
we let $|U|_h\equiv \sum_{i\in U} h_i$.
We define the \emph{tilted cutwidth of} $G$ as
\begin{eqnarray}
\label{eq:TiltedCutwidthDef} \Gamma(G;\uh) \equiv
\min_{S:\emptyset\to V} \max_{t\le n}\, \left[\cut(S_t,V\setminus
S_t)-|S_t|_h\right]\, .
\end{eqnarray}
Here the $\min$ is taken over all \emph{linear orderings} of the vertices
$i(1),\dots,i(n)$, with $S_t \equiv\{i(1),\dots,i(t)\}$. Note
that if for all $i$,
$h_i = 0$, the above is equal to the cutwidth of the graph.

Given a collection of subsets of $V$, $\Omega\subseteq 2^V$ such
that $\emptyset\in\Omega$, $V\not\in\Omega$,
we let $\dOmega$ be the collection of
couples $(S,S\cup\{i\})$ such that  $S\in \Omega$ and $S\cup
\{i\}\not\in\Omega$. We then define the \emph{tilted cut} of $G$ as
\begin{eqnarray}
\Delta(G;\uh) \equiv \max_{\Omega}
\min_{(S_1,S_2) \in\dOmega}\,
\max_{i=1,2}\left[\cut(S_i,V\setminus S_i)-|S_i|_h\right]\, ,\label{eq:TiltedCutDef}
\end{eqnarray}
the maximum being taken over \emph{monotone} sets $\Omega$
(i.e. such that $S\in \Omega$ implies $S'\in\Omega$ for all
$S'\subseteq S$). It thus coincide

It is known that, in the case $h_i=0$, the mixing time of
Glauber dynamics is at most exponential in the cutwidth of $G$
\cite{PeresEtAl}.
The following result provides a generalization to the case $h_i>0$
of interest here, in the limit of large $\beta$. Since
$\Gamma(G;\uh)$ (as well as $\Delta(G;\uh)$)
is decreasing in $\uh$, the upper bound
is smaller than the one for the $h_i=0$ case.
\begin{thm}\label{thm:Barrier}
Given an induced subgraph $F\subseteq G$, let $\uh^F$ be defined by
$h^F_i=h_i+|\di|_{G\setminus F}$, where  $|\di|_{G\setminus F}$
is the degree of $i$ in $G\setminus F$.
For reversible asynchronous dynamics we have
$\tau_+(G;\uh) = \exp\{2\beta\Gamma_*(G;\uh)+o(\beta)\}$, where
\begin{eqnarray}
\Gamma_*(G;\uh) =\max_{F\subseteq G}\Gamma(F;\uh^F)=
\max_{F\subseteq G}\Delta(F;\uh^F)\, .\label{eq:Barrier}
\end{eqnarray}
\end{thm}
Note that tilted cutwidth and tilted cut are dual quantities.
The former corresponds the maximal energy height along the lowest path
to the $+$ equilibrium. The latter is the lowest energy
along the highest separating set in the space of configurations.
A natural strategy for estimating $\Gamma_*(G;\uh)$
consists in lower bounding $\Delta(F;\uh^F)$ by exhibiting a
monotone set $\Omega\subseteq 2^{V(F)}$, and upper bounding
$\Gamma(F;\uh^F)$ by exhibiting a linear ordering of $V(F)$. The
above theorem shows that tilted cut and cutwidth coincide for the
`slowest' subgraph of $G$ and if the $h_i$'s are non-negative. The
hitting time is exponential in this graph parameter.


The two characterizations above are exact but it is highly non-trivial to
compute them. In the rest of this section, we will show how the
above theorem implies the known results for special classes of
graphs. Then, we relate tilted cutwidth to graph expansion and
derive a dichotomy between the hitting time on expanders versus
locally connected graphs. In the end, we show how to use algorithms
for sparsest cuts to find the approximately optimal linear ordering
as defined in tilted cutwidth.

The cases treated by Ellison are easily understood within the
present framework. In order to derive a lower bound for the complete
graph, with $h_i=h$ for all $i\in V$, one can restrict attention to
$F=G$ and for that graph define $\Omega$ to be the family of all
sets with cardinality at most $n/2$.
\begin{eqnarray}
 \Gamma^*(K_n; \uh) \geq \min_{|S| = n/2}   \left[\cut(S,V\setminus S)-|S|_h\right]
= (n-h)^2/4+O(n)\, .
\end{eqnarray}
The second example studied by Ellison is a $2k$-regular graph resulting
from connecting all vertices of distance at most $k$ in a cycle. In
that graph, the maximum is again achieved for $F=G$, and the natural
linear ordering of the cycle yields $\Gamma(G;h)\le 4k^2$.

It is also straightforward to recover the result of Young
\cite{youngdiffusion} from the above theorem.
Indeed, the hypotheses of \cite{youngdiffusion} are equivalent to
the existence of a sequence $S_1,\dots,S_T\subseteq V$ such that
$H(S_t) = \min_{S'\subseteq S_t}H(S')\le 0$ and $|S_i|\le k$.
By flipping vertices along this sequence and using the submodularity of
$H(\,\cdot\,)$, it follows that $\Gamma(F;\uh^F)\le k^2$.

%
%

\subsection{Relation to graph expansion}

The following Lemma links the isoperimetric function of $G$
(and its subgraphs) to the hitting time. It is particularly useful when
analyzing specific graph families.
\begin{lemma}\label{lemma:Isoperimetric}
For $\theta\in \reals$ define
$J(\theta) = [\theta-h_{\rm max},\theta+h_{\rm max}]$.
Assume that there exist constants $\alpha$ and
$\gamma<1$ such that
for any subset of vertices  $U\subseteq V$, and any $\theta$
such that there exists $S\subseteq U$ with $|S|_h\in J(\theta)$, we have
\begin{eqnarray}
\cut(S,U\setminus S) \le  \alpha\, |S|^{\gamma}\, ,\label{eq:IsoperimetricHypo}
\end{eqnarray}
for at least one such $S$. Then
$\Gamma_*(G;\uh)\le A(\alpha,\gamma,h_{\rm max})\, h_{\rm min}^{-1/(1-\gamma)}
\log\,  \max(2,h_{\rm min}^{-1})$.

Conversely, assume there exists   $U\subseteq V(G)$, such that
for $i\in U$, $|\partial i\cap (V\setminus U)|\le b$, and the subgraph induced
by $U$ is a $(\delta,\lambda)$ expander. Then
$\Gamma_*(G;\uh)\ge
(\lambda-h_{\rm max}-b)\lfloor \delta |U|\rfloor $.
\end{lemma}
In words, the hitting time is dominated by highly connected
subgraphs of $G$, that are loosely tied to the rest of the graph.
On the other hand,  an upper bound on the isoperimetric function leads to
upper bounds on the hitting time.

In order to gain some intuition we consider a few interesting graph models:
\begin{enumerate}
\item \emph{Finite-range $d$-dimensional networks.}
The graph $G$ is a $d$-dimensional range-$K$ network if we can
associate to each of its vertices $i\in V$ a position
$x_i\in\reals^d$ such that, $(1)$ whenever $(i,j)\in E$,
$d_{\Eu}(x_i,x_j)\le K$ (here $d_{\Eu}(\,\cdots\,)$ denotes
Euclidean distance); $(2)$ Any cube of volume $v$ contains at most
$2\,v$ vertices. We will also say that $G$ is \emph{embeddable} in
this case.
\item \emph{Small world networks.} Again, the vertices are those of a
$d$-dimensional grid of side $n^{1/d}$. Two vertices $i$, $j$ are connected
by an edge if they are nearest neighbors.
Further, each vertex $i$ is connected to $k$ other vertices $j(1)$,
$\dots$, $j(k)$ drawn independently with distribution
$P_i(j) = C(n) |i-j|^{-r}$.
\item \emph{Random regular graphs of degree $k$}.
\end{enumerate}

\begin{thm}\label{thm:Examples}
The following statements hold with high probability:

If $G$ is a $d$-dimensional finite-range graph, and $h_{\rm min}>0$,
then $\Gamma_*(G;\uh) = O(1)$.

If $G$ is a small world network with $r\ge d$, and
$h_{\rm max}\le k- d-5/2$, then
$\Gamma_*(G;\uh) = \Omega(\log n/\log\log n)$.

If $G$ is a small world network with $r< d$, and $h_{\rm max}$ is small enough,
then $\Gamma_*(G;\uh) =\Omega(n)$.

If $G$ is a random $k$-regular graph, and $h_{\rm max}<k-2$,
then $\Gamma_*(G;\uh) =\Omega(n)$.
\end{thm}
These qualitatively distinct behaviors correspond to different
mechanisms by which consensus spreads in these networks.
In finite-range networks, the process is initiated in a relatively compact
region taking value $+1$. If this is large enough
(which happens with positive probability), it spreads through the
whole graph.
This is possible because of the bias provided by $h_{\rm min}>0$.
Indeed the proof of this statement implies
an upper bound of the form $\Gamma(G;\uh)=
O(h_{\rm min}^{-(d-1)}\log(1/h_{\rm min}))$.

In small-world networks with $r\ge d$
the process is similar, but the spread
of $+1$'s is blocked in its very last stages by small, highly connected
regions of  size roughly $(\log n)$.
Finally,  small-world networks with $r< d$ and
random regular graphs are expanders and  convergence is extremely slow.

All the above statements take the form of a tradeoff
between how `well-connected' is $G$ and how biased is the dynamics
(the latter being measured by $h_{\rm min}$).
In the case of well-connected graphs it is not hard to prove
upper bounds on $\Gamma_*(G;\uh)$ for large enough $\uh$.
For instance, in the case of $k$-regular graphs $\Gamma_*(G;\uh)=O(1)$
if $h_{\rm min}\ge k$.

\subsection{Approximating tilted cut and tilted cutwidth}
%

The maximization over $\Omega$ in Eq.~(\ref{eq:TiltedCutDef}) for
computing tilted cut is highly non-trivial. Here we obtain a class
of lower bounds by restricting $\Omega$ to essentially subsets with
a given cardinality. The following result shows the `loss'
resulting from this restriction is bounded, under appropriate conditions.
On the other hand, it
implies that algorithms for computing sparse cuts find approximately
optimal orderings corresponding to a tilted cutwidth.

\begin{thm}\label{thm:Crux}
Assume that, for some $L_1, L_2$, with $L_2 \ge h_{\rm max}$ and for
every induced  subgraph $F\subseteq G$, we have
\begin{equation}
\label{tiltedexpansion}  \min_{|S|_{h}\in [L_1,L_2]}\, \left[
\cut(S, V(F)\setminus S) -|S|_{h^F}  \right]\le L_1\, ,
\end{equation}
where it is understood that $\emptyset\neq S\subseteq V(F)$. If, for
every subset of vertices $U$, with $|U|_h\le L_2$, the induced
subgraph has cutwidth upper bounded by $C$, then $\Gamma(G;4\uh) \le
C+L_1+L_2$.
\end{thm}
%


It is interesting to compare this result with the analysis of
contagion models \cite{Morris}. In that case contagion takes place
if there exists an ordering of the vertices $i(1)$, $i(2)$, \dots
such that, assuming $x_{i(1)}=+1$, $x_{i(2)}=+1$,\dots
$x_{i(t)}=+1$, the best response for $i(t+1)$ is strategy $+1$.
Theorem \ref{thm:Crux} allows to replace single vertices, by
`blocks' as long as they have bounded size and bounded cutwidth.

Assuming that a `good' path to consensus exists, can it be found
efficiently? By using a simple generalization of Feige and
Krauthgamer's \cite{FK} $O(\log^2 n)$ approximation algorithm for
finding the sparsest cut of a given cardinality, we have the
following
\begin{remark}\label{re:Coro}
If $G=(V,E)$ satisfies equation (\ref{tiltedexpansion}), it is
possible to find an ordering $i_1, i_2, \ldots, i_{n}$ of $V$ in
polynomial time so that for every $S_t = \{i_1, i_2, \ldots i_t\}$,
and $L=L_1+L_2+C$
$$\cut(S_t,V\setminus S_t) = O( |S_t|_h \log^2 n + L \log n).$$
\end{remark}

%
%
\subsection{Nonreversible and synchronous dynamics}

In this section we consider a general class of Markov dynamics over
$\ux\in \{+1,-1\}^V$. An element in this class is specified by
$p_{i,\beta}(y_i|\ux_{\di})$, with $p_{i,\beta}(+1|\ux_{\di})$ a
non-decreasing function of the number $\sum_{j\in \di}x_j$. Further
we assume that $p_i(+1|\ux_{\di}) \le e^{-2\beta}$ when
$h_i+\sum_{j\in \di}x_j<0$. Note that the synchronous Markov chain
studied in KMR \cite{KMR} and Ellison \cite{Ellison} is a special
case in this class.

Denote the hitting time of all (+1)-configuration in graph $G$ with
$\tau_+(G)$ as before.

\begin{propo}
Let $G(V,E)$ be a $k$-regular graph of size $n$ such that for
$\lambda, \delta > 0$, every $S \subset V, |S| \leq \delta n$ has
vertex expansion at least $\lambda$. Then for any noisy-best
response dynamics defined above, there exists a constant
$c=c(\lambda,\delta,k)$ such that $\tau_+(G;\uh)\ge \exp\{\beta c
 n\}$ as long as

$$ \lambda > \frac{3 k}{4} + \frac{ \max_i h_i}{2}. $$

\end{propo}
%
%
Note that random regular graphs satisfy the condition of the above
proposition as long as $h_i$'s are small enough.
The proof of the proposition is by simply considering the evolution
of one dimensional chain indicating the number of $+1$ vertices.

\begin{propo}
Let $G$ be a $d$-dimensional grid of size $n$ and constant $d \geq
1$. For any synchronous or asynchronous noisy-best response dynamics
defined above, there exists constant $c$ such that $\tau_+(G;\uh)\le
\exp\{\beta c\}$.
\end{propo}
The above proposition can be proved by a simple coupling argument
very similar to that of Young \cite{Young}. We will leave its
details to a more complete version of the paper. The above two
propositions show that for a large class of noisy best-response
dynamics including the one considered in \cite{Ellison}, the degrees
of vertices are not the key property dictating the rate of
convergence.

%
%
%
\section{Proofs}

\subsection{Theorem \ref{thm:Barrier}}

It is a basic result in the theory of reversible Markov chains with
exponentially small transition rates, that hitting time are related to
`energy barriers.'
\begin{lemma}\label{lemma:BasicMarkovChain}
Consider a Markov chain with state space $\cS$ reversible with respect to the
stationary measure $\mu_{\beta}(x) = \exp(-\beta H(x)+o(\beta))$,
and assume that, if $p_{\beta}(x,y)= \exp(-\beta V(x,y)+o(\beta))$.

Let $A= \{x:\, H(x)\le H_0\}$ be non-empty, and define the typical
hitting time for $A$ as in Eq.~(\ref{eq:TypicalHitting}), with $+$
replaced by $A$. Then $\tau_A= \exp\{\beta\tG_A+o(\beta)\}$
where
\begin{eqnarray}
\tG_A = \max_{z\not\in A} \min_{\omega:z\to A}\max_{t\le |\omega|-1}
\left[H(\omega_t)+V(\omega_{t},\omega_{t+1})
-H(z)\right]\,  \, ,\label{eq:BasicMarkovChain}
\end{eqnarray}
and the $\min$ runs over paths $\omega = (\omega_1,\omega_2,\dots,\omega_T)$
in configuration space such that $p_{\beta}(\omega_t,\omega_{t+1})>0$
for each $t$.
\end{lemma}
The proof can be obtained by building on known results, for instance
Theorem 6.38 in \cite{Olivieri}. These however typically apply to exit
times from local minima of $H(x)$. We provide a simple proof based on
spectral arguments in Appendix \ref{App:Barrier}.

For the sake of clarity, we split the proof of Theorem \ref{thm:Barrier}
in two parts: first  the characterization in terms of
tilted cutwidth (i.e. the first identity in Eq.~(\ref{eq:Barrier}));
then the one in terms of tilted cut (second
identity in Eq.~(\ref{eq:Barrier})).

\begin{proof}(Theorem \ref{thm:Barrier}, Tilted cutwidth).
Notice that Glauber dynamics satisfies the
hypotheses of Lemma \ref{lemma:BasicMarkovChain}, with
$H(\ux) = H(\ux)$ given by Eq.~(\ref{eq:EnergyAsynchron}).
In this case, for any allowed transition $\ux\to \uy'$,
 $H(\ux)+V(\ux,\uy)=\max(H(\ux),H(\uy))$. As a consequence, we can
drop the factor $V(\cdots )$ in Eq.~(\ref{eq:BasicMarkovChain}).
We thus obtain $\tau_+ = \exp(\beta\max_{\uz} \tG_+(\uz)+o(\beta))$ where
\begin{eqnarray}
\tG_+(\uz) =  \min_{\omega:\uz\to \up}\max_{t\le |\omega|-1}
\left[H(\omega_t)-H(\uz)\right]\, .
\end{eqnarray}
An upper bound is obtained by restricting the minimum to monotone
paths. It is not hard to realize that the
result coincides with $2\Gamma(F;\uh^F)$ where $F$ is the subgraph
induced by vertices $i$ such that $z_i=-1$.
It is far less obvious  that the optimal path
can indeed be taken to be monotone.

It is convenient to use the representation of the path
$\omega=(\ux_0=\uz,\ux_1,\dots,\ux_{|\omega|-1}=\up)$ as a sequence
of subsets of vertices: $\omega =
(S_0=S,S_1,\dots,S_{|\omega|-1}=V)$. We will consider a more general
class of paths whereby $ S_t\setminus S_{t-1} = \{v\}$ or  $ S_t
\subset S_{t-1}$, and let $G(\omega) = \max_t [H(S_t)-H(S_0)]$.

Let us start by considering the optimal initial configuration
We claim that  if $B \in \arg\max_{S} \min_{\omega:S \rightarrow V} G(\omega)$
is such an optimal configuration, then for every $A \subset B$,
$H(A) \ge H(B)$.
Indeed, suppose $H(A) < H(B)$. By prepending $B$ to any path $\omega:A
\rightarrow V$, we obtain a path $\omega':B \rightarrow V$ with
$G(\omega')< G(\omega)$. Therefore $\min_{\omega':B \rightarrow V} G(\omega')  <
\min_{\omega:A \rightarrow V} G(\omega)$ which is a contradiction.

Among all paths that achieve the optimum, choose the path $\omega$
that minimizes the  potential function $f(\omega) = |\omega|^2 |V| -
\sum_{S_i \in \omega} |S_i|$.  Intuitively, $f$  puts a very high weight
on shorter paths and then paths with larger sets. We will prove that,
with this choice, $\omega$ is monotone.

For the sake of contradiction, suppose $\omega$ is
not monotone. Let $S_{k}$ be the set with the smallest index such
that $ S_{k+1} \subset S_{k}$. Partition $ S_{k} \setminus S_{k+1}$
into two subsets $R = (S_{k} \setminus S_{k+1}) \cap S_0$ and $T =
(S_{k} \setminus S_{k+1}) \setminus S_0$. Without loss of
generality assume that for $1 \leq i \leq k$, $S_i = \{1, 2, \cdots
i\}\cup S_0$. Let $v_1 \leq v_2 \cdots \leq v_t$ be the elements of $T$
in the order of their appearance in $\omega$.

For a subset $A \subset T$, and $i \leq k$ define the marginal value
of subset $A$ at position $i$ to be $M(A, i) = H(S_i
\setminus A)-H(S_i)$. Since $H$ is submodular, $M(A, i)$ is
non-decreasing with $i$ as long as $A \subset S_i$.
Because of our claim about the initial condition, we have, in particular,
\begin{equation}
\label{eq:s0}
 M(R, 0) = H(S_0) - H(S_0 \setminus R) \geq 0\, .
\end{equation}

The crucial lemma below is proved in Appendix \ref{app:TwoCases}.
\begin{lemma}\label{lemma:TwoCases}
One of the following two statements is correct:
Case \text{(I)} There exists a subset $T' \subset T$ such that for all $i$,
 $M(T', i) \leq 0$;
Case \text{(II)} $M(T \cup R, k) \geq 0$.
\end{lemma}

We are now ready to finish the proof. Suppose the
first statement of the lemma is correct.  We construct a new path $\omega'$ by
removing the vertices of $T'$ from the sequence $1,2,\cdots, t$ in
the beginning of $\omega$ and also removing $T'$ from $T$. Since $\omega'$ is
shorter than $\omega$, we only need to argue that $G(\omega') \le
G(\omega)$. This is obvious because for every $i \leq k$,
$ H(S_i \setminus T')- H(S_i) = M(T', i)  \leq 0$.

In the second case, we construct another path by changing $S_{k+1}$.
First note that since $\omega$ is minimizing the potential function,
$S_{k+2} = S_{k+1} \cup \{v\}$ for some $v$ that is not in $S_k$.
Now note that by replacing $S_{k+1}$ with $S_k \cup \{v\}$ we obtain
a path with a higher value of the potential function and at most the
same barrier. This is because
\begin{eqnarray}
H(S_{k+1} \cup \{v\})- H(S_{k} \cup \{v\})   \geq H(S_{k+1})- H(S_{k})
 = M(T \cup R, k) \geq 0\, .
\end{eqnarray}

\end{proof}

The second part of the proof exploits the well known fact
that Glauber dynamics is monotone
for the Ising model. Given initial conditions
$\ux(0)$ and $\ux'(0)\succeq \ux(0)$, the corresponding evolutions
can be coupled in such a way that $\ux'(t)\succeq \ux(t)$ after any number
of steps.

\begin{proof} (Theorem \ref{thm:Barrier}, Tilted cut).
By monotonicity of Glauber dynamics $\Gamma_*(G;\uh)\ge \Gamma_*(F;\uh^F)$
for any induced subgraph $F\subseteq G$. Theorem
\ref{lemma:BasicMarkovChain} implies $\Gamma_*(F;\uh^F)\ge \Delta(F;\uh^F)$:
indeed given a path $\omega=(S_0,S_1,\dots ,S_{|\omega|-1}=V)$
this must have at least one step in $\dOmega$. Hence
$\Gamma_*(G;\uh)\ge \max_F \Delta(F;\uh^F)$.

We need to prove $\Gamma_*(G;\uh)\le \Delta(F;\uh^F)$ for at least
one induced subgraph $F$. Fix $F$ to be a subgraph which achieves the maximum
in Eq.~(\ref{eq:Barrier}) (i.e. $\arg\max\Gamma(F;\uh^F)$).
Notice that, to leading exponential order, the hitting
time in $F$ is the same as in $G$, i.e.
$\Gamma_*(F;\uh^F) = \Gamma_*(G;\uh)$.

Let $p_{\beta}(\ux,\uy)$ be the transition probabilities
of Glauber dynamics on $F$, and $p_{\beta}^+(\ux,\uy)$ the kernel restricted
to $\{ +1,-1\}^{V(F)}\setminus\{\up\}$.
By this we mean that we set $p_{\beta}^{+}(\ux,\up)=p_{\beta}^{+}(\up,\uy)=0$.
Denote by $P^+_{\beta}$ the matrix with entries $p_{\beta}^{+}(x,y)$
and by $\psi_0$ its eigenvector with largest eigenvalue.
By Perron-Frobenius Theorem, we can assume $\psi_0(\ux)\ge 0$.
We claim that $\psi_0(\ux)$ is monotonically decreasing in $\ux$.
Indeed consider the transformation $\psi\mapsto T(\psi) \equiv P^+_{\beta}\psi/
||P^+_{\beta}\psi ||_{2,\mu}$. This is a continuous mapping from the
set of unit vectors in
$L^2(\mu)$ onto itself. Further, if $\psi$ is monotone and
non-negative, $T(\psi)$ is monotone an non-negative as well
(the first property follows from monotonicity of the dynamics).
The set of non-negative and monotone unit vectors in $L^2(\mu)$
is homeomorphic to a simplex.
By Brouwer fixed point theorem, $T$ has at least one fixed point that is
non-negative and monotone, which therefore coincides with $\psi_0$
by Perron-Frobenius.

Lemmas \ref{lemma:SimpleSpectral} and \ref{lemma:Eigenvectors}
imply that there exists $\Omega=\{x\in\cS:\, \psi_0(\ux)> b\}$,
such that
\begin{eqnarray}
\tau_+(F;\uh^F)\le  C_n(1+\beta)\, \frac{\sum_{\ux\in\Omega}\mu(\ux)}
{\sum_{(\ux,\uy)\in\dOmega}\mu(\ux)p^+_\beta(\ux,\uy)}\, .
\end{eqnarray}
for some $\beta$-independent constant $C_n$. Using
$\tau_+(F;\uh^F)=\exp\{2\beta\Gamma_*(F;\uh^F)+o(\beta)\}$ and
the large $\beta$ asymptotics of $\mu(\ux)$, $p^+_\beta(\ux,\uy)$
we get
\begin{eqnarray}
\Gamma_*(F;\uh^F)\le  \min_{(S_1,S_2) \in\dOmega}\,
\max_{i=1,2}\left[\cut(S_i,V\setminus S_i)-|S_i|_h\right] + o_{\beta}(1)\, .
\end{eqnarray}
Since $\psi_0(\ux)$ is monotone, $\Omega$ is monotone as well
and therefore the last inequality implies the thesis.
\end{proof}

\vspace{-0.5cm}

%
%
\subsection{Theorem \ref{thm:Examples}}

\begin{proof}(Lemma \ref{lemma:Isoperimetric}).
By Theorem  \ref{thm:Barrier}, it is sufficient to find an upper bound for $\Gamma(\tilde{F};\uh^{\tilde{F}})$
for every induced subgraph $\tilde{F}$. By monotonicity of
$\Gamma(\tilde{F};\uh)$ with respect to $\uh$,  $\Gamma(\tilde{F};\uh^{\tilde{F}})
\le  \Gamma(\tilde{F};\uh)$.
We will upper bound $\Gamma(\tilde{F};\uh)$ by showing
Eq.~(\ref{tiltedexpansion}) holds for any induced subgraph $F\subseteq \tilde{F}$.

First notice that, for any $U$ and for any $\theta$,
 there exists $S\subseteq U$
such that $|S|_h\in J(\theta)$ and
\begin{eqnarray}
\cut(S,U\setminus S)-\frac{1}{4}|S|_h\le \alpha h_{\rm min}^{-\gamma}|S|_h^{\gamma}
- \frac{1}{4}|S|_h\le
A'(\alpha,\gamma) \, h_{\rm min}^{-\gamma/(1-\gamma)}
\, ,\label{eq:IsoperimetricBound}
\end{eqnarray}
where $A'(\alpha,\gamma) = \max( \alpha x^{\gamma}-x/4 :\, x\ge 0)$.
Take $L_1=A'(\alpha,\gamma) \, h_{\rm min}^{-\gamma/(1-\gamma)}$ and
$L_2= L_1+2h_{\rm max}$.
By Eq.~(\ref{eq:IsoperimetricBound})
\begin{eqnarray}
\min_{|S|_h\in[L_1,L_2]}\;\left[\cut(S,V(F)\setminus S)-\frac{1}{4}|S|_h\right]
\le L_1\, \nonumber.
\end{eqnarray}
Finally the cutwidth of any set $S$ with $|S|_h\le L_2$
is upper bounded by $\alpha|S|^{\gamma}\log |S|$ (using \cite{LR} and
Eq.~(\ref{eq:IsoperimetricHypo})) which is at most
$C= A''(\alpha,\gamma,h_{\rm max}) \, h_{\rm min}^{-1/(1-\gamma)}
\log\max(2,h_{\rm min}^{-1})$.
The thesis thus follows by applying Theorem \ref{thm:Crux}.

To prove the lower bound we use
Theorem \ref{thm:Barrier} again. Let $F$ be the subgraph induced by $U$.
By monotonicity of $\Delta(G;\uh)$ with respect to $\uh$,
for $t=\lfloor \delta|U|\rfloor$, we have
\begin{eqnarray}
\Delta(F;\uh^F) \ge \Delta(F;h_{\rm max}+k)\ge \min_{|S|= t}
\left[\lambda|S|-(h_{\rm max}+k)|S|\right]\, \nonumber.
\end{eqnarray}
which implies the thesis.
\end{proof}

We notice in passing that the estimates in the second part of this proof
could be improved by using more specific arguments instead of
directly applying Theorem \ref{thm:Barrier}.

For the proof of theorem \ref{thm:Examples}, we need to estimate
the isoperimetric function of finite range $d$-dimensional graphs.
This can be done by an appropriate relaxation.

Given a function $f:V\to \reals$, $i\mapsto f_i$, and a set of
non-negative weights $w_i$, $i\in  V$, we define
\begin{eqnarray}
||f||_w^2 \equiv \sum_{i\in V}w_i\, f_i^2\, ,\;\;\;\;\;\;
||\nabla_G f||^2 \equiv \sum_{(i,j)\in E}|f_i-f_j|^2\, .
\end{eqnarray}
We then have the following generalization of Cheeger inequality.
\begin{lemma}\label{lemma:Cheeger}
assume there exists two vertex sets $\Omega_1\subseteq\Omega_0\subseteq V$
and a function $f:V\to \reals$ such that:
$(1)$ $f_i \ge |f_j|$ for any $i\in \Omega_1$ and any $j\in V$;
$(2)$ $f_i = 0 $ for $i\in V\setminus\Omega_0$;
$(3)$ $L_1\le |\Omega_1|_w\le |\Omega_0|_w\le L_2$; $(4)$
$||\nabla_G f||^2\le \lambda \, ||f||_h^2$.
Then there exists $S\subseteq V$ with $L_1\le |S|_w\le L_2$
\begin{eqnarray}
\cut(S,V\setminus S)\le \sqrt{4\lambda\,\max_{i\in V}\{|\di|/h_i\}}\;\; |S|_h\, .
\end{eqnarray}
\end{lemma}

The proof of this Lemma is deferred to Appendix \ref{sec:CheegerProof}.
\begin{proof}(Theorem \ref{thm:Examples})
\emph{Finite-range $d$ dimensional networks.} We need to prove that,
for each induced subgraph $G'$, $\Gamma(G';\uh^{G'})=O(1)$. By
Theorem \ref{thm:Crux}, it is sufficient to show that, for any
induced and connected  subgraph $F$, there exists a set $S$ of
bounded size such that $\cut(S,V(F)\setminus S)-
\frac{1}{4}|S|_{(h)^F}\le 0$, with $h'_i=h_i/4$. If the original
graph is embeddable, any induced subgraph is embeddable as well.
Since $h^F_i\ge h_i$, the thesis follows by proving that for any
embeddable graph $G$, we can find a set of vertices $S$ of bounded
size with $\cut(S,V\setminus S)\le |S|_{h/4}$.

We will construct a function $f$ with bounded support such that
$||\nabla_G f||^2\le \lambda ||f||^2$ with $\lambda =
\min_{i\in V}\{\frac{h_i}{16|\di|}\}$. In order to achieve this goal,
consider the $d$-dimensional of $G$ and partition $\reals^d$
in cubes $\cC$ of side $\ell$ to be fixed later. Denote by $\cC_0$ the cube
maximizing $\sum_{i:x_i\in \cC} h_i$, and let
 $C_j$, $j=1,\dots 3^{d}-1$ be the adjacent cubes.
Let $f_i = \varphi(x_i)$, where for $x\in\reals^d$, we have
\begin{eqnarray}
\varphi(x) = \left[1-\frac{d_{\Eu}(x,\cC)}{\ell}\right]_+\, .
\end{eqnarray}
Notice that $|\nabla\varphi(x)|\le 1/\ell$  and
$|\nabla\varphi(x)|>0$ only if $x\in\cC_j$, $j=1,\dots 3^{d-1}$.
Since $|f_i-f_j|\le |\nabla\varphi|\;\; ||x_i-x_j||$ we have
\begin{eqnarray}
||\nabla_G f||^2&\le& \left(\frac{K}{\ell}\right)^2\sum_{i\in V}|\di|\;
\ind(x_i\in\cup_{j=1}^{3^d-1} \cC_j)\le
\left(\frac{K}{\ell}\right)^2\max_{i\in V}
\{|\di|/h_i\}\sum_{i\in V} h_i\;
\ind\Big(x_i\in\cup_{j=1}^{3^d-1} \cC_j\Big)\nonumber\\
& \le &
3^d \left(\frac{K}{\ell}\right)^2
\max_{i\in V}\{|\di|/h_i\}\sum_{i\in V} h_i\;
\ind\Big(x_i\in\cC_0\Big)\le 3^d \left(\frac{K}{\ell}\right)^2
\max_{i\in V}\{|\di|/h_i\}||f||^2_h \, .
\end{eqnarray}
The thesis follows by choosing $\ell = 2^{d+2}K \max_{i\in
V}\{|\di|/h_i\}$. \vspace{0.1cm}

\emph{Small world networks with $r\ge d$.} Let $U$ be a subset of
vertices forming a cube of side $\ell$, and $G_U$ a $(\ve,k-5/2)$,
$k$-regular expander with vertex set $U$. Such a graph exists for
all $\ell$ large enough and $\ve$ small enough by
\cite{ExpanderExist}. Call $A_U$ the event that the subgraph induced
by long-range edges in $U$ coincides with $G_U$, and no long-range
edge from $i\in V\setminus U$ is incident on $U$.

Under $A_U$, the subgraph $G_U$ satisfies the hypotheses of
Lemma \ref{lemma:Isoperimetric}, second part, with $b= d$.
Therefore $\Gamma_*(G;\uh)\ge (k-5/2-h_{\rm max}-d)\lfloor\ve\ell^d/4\rfloor$.
The thesis thus follows if we can prove the existence of $U$
with volume $\ell^d= \Omega(\log n/\log\log n)$ such that $A_U$ is true.

Fix one such cube $U$. The probability that the long range edges inside
$U$ induce the expander $G_U$ is larger than
$(C(n)\ell^{-r})^{k\ell^{d}} $. On the other hand, for any vertex
$i\in U$, the probability that no long range edge from $V\setminus U$
is incident on $U$ is lower bounded as
\begin{eqnarray}
\prod_{j\in V\setminus i} \left[1-C(n)|i-j|^{-r}\right]^{k}\ge
\exp\Big\{-3k\, C(n)\, \sum_{j\in V\setminus i} |i-j|^{-r}\Big\}\nonumber
\end{eqnarray}
where we used the lower bound $1-x\ge e^{-3x}$ valid for all
$x\le 1/2$, together with the fact that $C(n)\le 1/2d$ (which follows by
considering the $2d$ nearest neighbors). From the definition of
$C(n)$, the last expression is lower bounded by $e^{-3k}$, whence
\begin{eqnarray}
\prob\{A_U\} \ge \left[C(n)e^{-3} \ell^{-r}\right]^{k\ell^d}\,\nonumber .
\end{eqnarray}

Let $S$ denote a family of $(n/\ell^d)$ disjoint subcubes, and
denote by $N_S$ the number of such subcubes for which property $A_U$
holds. Then $\E[N_S]=(n/\ell^d) \prob\{A_U\}$. Using the above lower
bound together with the fact $C(n)\ge C_{r,d}>0$ for $r>d$ and
$C(n)\ge C_{*,d}/\log n$ for $r=d$, it follows that there exists
$a,b>0$ such that $\E[N_S] =\Omega (n^a)$ if $ell^d\le b\log
n/\log\log n$.

The proof if finished by noticing that, for $U\cap U'=\empty$,
$\prob\{A_U\cap A_{U'}\}\le \prob\{A_U\cap A_{U'}\}$,
whence $\Var(N_S)\le \E[N_S]$. The thesis follows applying
Chebyshev inequality to $N_S$.
\vspace{0.1cm}

\emph{Small world networks with $r< d$.} It is proved in \cite{Abie}
that these graphs are with high probability expanders. The thesis
follows from Lemma \ref{lemma:Isoperimetric}. \vspace{0.1cm}

\emph{Random regular graphs.} It is well known that a random $k$-regular
graph is with high probability a $k-2-\delta$ expander for all
$\delta>0$ \cite{ExpanderExist}. The thesis follows again from
Lemma \ref{lemma:Isoperimetric}.
\end{proof}

{\bf Acknowledgement.} We would like to thank Daron Acemoglu, Glenn
Ellison,  Fabio Martinelli, Roberto Schonmann, Eva Tardos, Maria Eulalia Vares,
and Peyton Young for helpful discussions and pointers to the literature.

%
%
\bibliographystyle{amsalpha}

%
%
\appendix

\section{Proof of Lemma \ref{lemma:Cheeger}}\label{sec:CheegerProof}

Assume without loss of generality that $\max\{|f_i|\;:\;\; i\in V\}=1$,
whence $f_i=1$ for $i\in \Omega_1$.
We use the same trick as in the proof of the standard Cheeger inequality
\begin{eqnarray}
||\nabla_G f||^2 =\sum_{(i,j)\in E}(f_i-f_j)^2\ge
\frac{\left(\sum_{(i,j)\in E}|f_i^2-f_j^2|\right)^2}{\sum_{(i,j)\in E}
(f_i+f_j)^2}\, .
\end{eqnarray}
The denominator is upper bounded by
\begin{eqnarray}
4\sum_{i\in V} |\di|\, f_i^2\le
4\max\left|\frac{|\partial i|}{h_i}\right|\;\; ||f||_h^2\, .
\end{eqnarray}
The argument in parenthesis at the numerator is instead equal to
\begin{eqnarray}
\sum_{(i,j)\in E}\int_0^1 \left|\ind(f^2_i>z)-\ind(f^2_j>z)\right| \de z
=\int_{0}^1\cut(S_z,V\setminus S_z)\, \de z
\end{eqnarray}
where $S_z = \{i\in V:\; f_i^2>z\}$.
The quantity above is lower bounded by
\begin{eqnarray}
&&\min_{z\in [0,1]}
\frac{\cut(S_z,V\setminus S_z)}{|S_z|_h}
\int_{0}^1|S_z|_h\, \de z =
\min_{z\in [0,1]}
\frac{\cut(S_z,V\setminus S_z)}{|S_z|_h}
\; ||f||_h \, .
\end{eqnarray}
Let $S= S_{z_*}$ where $z_*$ realizes the above minimum (the
function to be minimized is piecewise constants and right continuous
hence the minimum is realized at some point).
 Notice that
$\Omega_1\subseteq  S_z\subseteq \Omega_0$ for all $z\in[0,1]$,
and thus we have in particular $L_1\le |S|_w\le L_2$. Further, form the above
\begin{eqnarray}
\lambda\ge\frac{||\nabla_G f||^2}{||f||_h^2}\ge \frac{1}{4}\min\left|\frac{h_i}{|\di|}
\right|
\left\{\frac{\cut(S,V\setminus S)}{|S|_h}\right\}^2
\end{eqnarray}
which finishes the proof.
\endproof

%
%
\section{Hitting times at low temperature:
proof of Lemma \ref{lemma:BasicMarkovChain}}\label{App:Barrier}

We consider a general setting of Lemma \ref{lemma:BasicMarkovChain}:
a discrete time Markov chain with state space $\cS$,
transition probabilities $p_{\beta}(x,y)$, reversible with respect
to the stationary distribution $\mu(x)$.
Given $A\subseteq \cS$ define
$p_{\beta}^A(x,y) = p_{\beta}(x,y)$ if $x,y\in \cS\setminus A$
and $p_{\beta}^A(x,y)= 0$ otherwise. Notice by reversibility the
eigenvalues of $p_{\beta}^A$ are real, and smaller than $1$.
We assume that $p_{\beta}^A$ is irreducible and aperiodic.

The lower bound in the next lemma is due to Donsker and
Varadhan \cite{Donsker}: we nevertheless propose an
elementary proof.
\begin{lemma}\label{lemma:SimpleSpectral}
If $1-\lambda_{0,A}$ is the largest eigenvalue of $p_{\beta}^A$, then
\begin{eqnarray}
\frac{1}{\log(1/(1-\lambda_{0,A}))}\le \tau_A\le
\frac{1}{\log(1/(1-\lambda_{0,A}))}\left\{1+\frac{1}{2}\max_{x\in \cS\setminus A}
\log\frac{1}{\mu(x)}\right\}\,\nonumber .
\end{eqnarray}
\end{lemma}
\begin{proof}
Let $P_A$ denote the matrix with entries $p_{\beta}^A(x,y)$, and $f(x)$
be the characteristic function of $\cS\setminus A$. Then
$\prob_x\left\{T_A>t\right\} = P_A^tf(x)$, whence
\begin{eqnarray}
\sqrt{\mu(x)}\,\prob_x\{T_A>t\} \le \sqrt{\sum_x\mu(x)\prob_x\{T_A>t\}^2}
= ||P_A^tf||_{\mu,2}\le (1-\lambda_{0,A})^t\, \nonumber,
\end{eqnarray}
which proves the upper bound. To prove the lower bound, let $\psi_0(x)$
denote the eigenvector of $P_A$, with eigenvalue $\lambda_{0,A}$ and notice
that by Perron-Frobenius theorem, it has non-negative entries. Therefore
\begin{eqnarray}
\max_x\prob_x\{T_A>t\}\, (\psi_0,f)_{\mu} \ge \sum_x \mu(x)\psi_0(x)\prob_x\{T_A>t\} = (1-\lambda_{0,A})^t(\psi_0,f)\,\nonumber .
\end{eqnarray}
\end{proof}

\begin{proof}(Lemma \ref{lemma:BasicMarkovChain}).
Due to Lemma \ref{lemma:SimpleSpectral}, it is sufficient to prove
that $\lambda_{0,A} = \exp\{-\beta\tG_A+o(\beta)\}$. To this end
we use the well known variational characterization of eigenvalues
\begin{eqnarray}
\lambda_{0,A} = \inf_{\varphi}\;\frac{\Dir(\varphi)}{\E(\varphi^2)}\, ,\;\;\;\;\;
\;\;\;\;\;\;\;\Dir(\varphi) \equiv \frac{1}{2}
\sum_{x,y}\mu(x)p_{\beta}(x,y)(\varphi(x)-\varphi(y))^2\, .
\label{eq:Dirichlet}
\end{eqnarray}
Here the $\inf$ is taken over functions non-vanishing
functions $\varphi:\cS\setminus A\to \reals$.

A lower bound can be obtained by comparison. More precisely, for
each $z\in \cS\setminus A$, let $\omega^{(z)}$ be a path or allowed
transition from $z$ to $A$. Proceeding along the lines of
\cite{JerrumSinclair,DiaconisSaloffCoste}, one obtains that
$\lambda_{0,A}\ge 1/\max_{x,y} C(x,y;\omega)$, where, for each
allowed transition $x\to y$, we defined the associated congestion as
\begin{eqnarray}
C(x,y;\omega) = \frac{1}{\mu(x)p_{\beta}(x,y)} \sum_{z:\omega^{(z)}\ni (x,y)}
\mu(z)|\omega^{(z)}|\, \nonumber .
\end{eqnarray}
The thesis then follows by choosing the path
$\omega^{(z)}$ in such a way to achieve the minimum in
Eq.~(\ref{eq:BasicMarkovChain}) and taking the limit $\beta\to\infty$.

To get an upper bound, define the boundary $\partial B$
of a configuration $B$, as the subset of couples $(x,y)$ such that
$p_{\beta}(x,y)>0$ and $x\in B$, while $y\not\in B$.
Notice that from Eq.~(\ref{eq:BasicMarkovChain})
it follows that there exists a set $B\subseteq \cS\setminus A$ such that
\begin{eqnarray}
\tG_A = \min_{(x,y)\in\partial B} [H(x)+V(x,y)]-\min_{z\in B} H(z)\, \nonumber .
\end{eqnarray}
The proof is completed by taking $\varphi$ in Eq.~(\ref{eq:Dirichlet})
to be the characteristic function of $B$.
\end{proof}

%
%
\section{Proof of Lemma \ref{lemma:TwoCases}}\label{app:TwoCases}

Construct the following partitioning of $T$ into $T_1 = \{ v_1, v_2,
\cdots  v_{i_1-1}\}$, $T_2 = \{v_{i_1}, v_{i_1+1}, \cdots
v_{i_2-1}\}$ $\cdots T_r = \{v_{i_{r-1}} \cdots v_k \}$ in such a
way that for every $T_j = \{v_{i_{j-1}}, \cdots v_{i_j - 1}\}$ and
$i_{j-1} < l < i_{j}$, $M(T_j, v_l - 1) = M(\{v_{i_{j-1}}\cdots v_{l-1}\}, v_l -1) < 0$ and for $l = i_j$,
$M(T_j, v_l - 1) \geq 0$.

Such a partition can be obtained the following way. Start with $j= 1$ and iteratively add $v_i$'s to the current set $T_j$. If $M(T_j, v_i - 1) \geq 0$, increment $j$ and add $v_i$ and the next vertices to the new subset.

Let $T_ r = \{v_s, \cdots, v_t\}$ be the last subset in the above
sequence. We claim that if $M(T_r, k) < 0$ then $M(T_r, i) < 0$ for
all $i \geq v_s$. For every $s \leq j \leq t $ and every $i$ between
$v_j$ and $v_{j+1}$  by supermodularity $M(T_r, i) = M(\{v_l, \cdots
v_j\}, i) \leq M(\{v_l, \cdots v_j\}, v_{j+1} -1) < 0$. The same
argument goes for $v_t \leq i \leq k$. In that case the lemma is
correct for $T' = T_r$.

If $M(T_r, k) \geq 0$, we will show that the second statement of the
lemma is true. For that, we need to write the $H$ function for all
sets $T_1, \cdots T_r$ explicitly.  For a set $T_j$ and $l = i_j$
\begin{eqnarray}
M(T_j, v_l - 1) = 2 \left[ \cut(T_j, \{1, 2, \cdots v_l - 1 \}) -
\cut(T_j, \{v_l, v_l + 1, \cdots n\}) +
 \sum_{i \in T_j} h_i \right] \geq 0\, .
\end{eqnarray}

One can write a similar equation $j = l$ by replacing $v_l - 1$ with
$k$. Equation (\ref{eq:s0}) gives a similar inequality for $R$.
Adding up these inequalities for all $j$ and $R$ and noting that the
contribution of every edge with both ends in $\cup_{j} T_j \cup R$ cancels out, we get
\begin{eqnarray}
 M(T \cup R, k) \geq \sum_{j=1}^{l-1} M(T_j, v_{i_j} -1)  + M(T_l, k) + M(R, 0) \geq 0.
\end{eqnarray}
\endproof

\section{Proof of Theorem \ref{thm:Crux}}

\begin{proof}(Theorem \ref{thm:Crux}).
Partition $V$ into subsets $R_1, R_2, \cdots, R_l$ by letting
$V_0\equiv V$ and defining recursively
\begin{eqnarray}
 R_t = \argmin_{S\in \Omega_t} \{ \cut(S, V_t \setminus S) -
|S|_{h^{V_t}} \}\nonumber
\end{eqnarray}
where $V_t = V \setminus \cup_{s=1}^{t-1} R_s$ and $\Omega_t$ is the
set of all subsets $S\subseteq V_t$ such that $L_1\le |S|_{h}\le
L_2$.
 With an
abuse of notation, we wrote $\uh^{V_t}$ for $\uh^{G(V_t)}$ ($G(V_t)$
being the subgraph induced by $V_t$). Explicitly, for any $j\in
V_t$, $(h^{V_t})_j = h_j+ |\partial j|_{V\setminus V_t}$.

Continue this process until no such set $S$ can be found, and let
$R_l=V_l$ be the residual set. Notice that, since $L_2\ge h_{\rm
max}$, we necessarily have $|R_l|_h<L_1$. By applying
Eq.~(\ref{tiltedexpansion}) to $F=G(V_t)$, we have
\begin{equation}
\label{eq:ri} \cut(R_t, V_t \setminus R_t) \le  |R_t|_{h^{V_t}} +
L_1 \leq |R_t|_{h^{V_t}}+|R_t|_{h} = |R_t|_{2h}+\cut(R_t,V\setminus
V_t)\, .
\end{equation}
Notice that $\cut(R_t,V_t\setminus R_t)-\cut(R_t,V\setminus V_t)=
\cut(\cup_{s=1}^tR_s,V_{t+1})-\cut(\cup_{s=1}^{t-1}R_s,V_t)$. By
summing up this relation, we have, for all $1 \le t  < l$,
\begin{equation}
 \cut(\cup_{s=1}^t R_s, V \setminus \cup_{s=1}^t R_s)
\le \sum_{s=1}^{t}|R_s|_{2h}
  = |\cup_{s=1}^t R_s|_{2h} \nonumber.
\end{equation}

For each $R_t$, consider a linear arrangement of the induced
subgraph that achieves its cutwidth. Construct a linear arrangement
of $V$ by concatenating the above linear arrangement of each $R_t$
in the order $t=1, 2, \dots,l$. We will show that this ordering
gives us the desired upper bound on the tilted cutwidth of $G$. Let
$S = \cup_{s=1}^{t-1} R_s \cup R$ where $R \subset R_t$ for some $t$
between $1$ and $l$. Then
\begin{eqnarray}
\cut(S, V \setminus S) & \leq & \cut(\cup_{s=1}^{t-1} R_s,
V\setminus \cup_{s=1}^{t-1} R_s)+\cut(R_t, V\setminus V_t)
  + \text{cutwidth}(R_t)
\nonumber\\ & \le &  \cut(\cup_{s=1}^{t-1} R_s, V\setminus
\cup_{s=1}^{t-1} R_s) + \cut(R_t,V\setminus V_t)+|R_t|_h+L_1+C
\nonumber\\ & \le & 2\; \cut(\cup_{s=1}^{t-1} R_s, V\setminus
\cup_{s=1}^{t-1} R_s) +L_1+L_2+C \nonumber\\ & \le &
2|\cup_{s=1}^{t-1} R_s|_{2h} +L_1+L_2+C\, \nonumber.
\end{eqnarray}
\end{proof}

%
%
\section{Eigenvectors and barriers}

As in the last appendix, we consider here a general Markov
chain with state space $\cS$, and let $A\subseteq \cS$ a subset of
configurations.

\begin{lemma}\label{lemma:Eigenvectors}
Let $\psi_0 :\cS\to\reals$ be the unique eigenvector of
$P_A$ with eigenvalue $1-\lambda_{0,A}$ and assume (without loss of
generality by Perron-Frobenius theorem) $\psi_0(x)\ge 0$.
Then there exists $b\ge 0$ such that, letting
$B=\{x\in\cS:\, \psi_0(x)> b\}$, we have
\begin{eqnarray}
\frac{1}{|\cS|}\,
\frac{\sum_{(x,y)\in\partial B}\mu(x)p_{\beta}(x,y)}{\sum_{x\in B}\mu(x)}
\le \lambda_{0,A}\le
\frac{\sum_{(x,y)\in\partial B}\mu(x)p_{\beta}(x,y)}{\sum_{x\in B}\mu(x)}
\end{eqnarray}
\end{lemma}
\begin{proof}
The upper bound follows immediately by substituting
$\varphi(x) = \ind(x\in B)$ in the variational
principle (\ref{eq:Dirichlet}).

In order to prove the lower bound,
let $0=\psi^{(0)}<\psi^{(1)}\le \cdots\le \psi^{(N)}$ be the points in
the image of $\psi_0(\,\cdot\,)$ (obviously $N\le \cS$).
For any $(x,y)$ such that $\psi_0(x) = \psi^{(i)}$,
$\psi_0(y) = \psi^{(j)}$, with $i<j$, we have
$(\psi_0(x)-\psi_0(y))^2\ge \sum_{l=i}^{j-1}(\psi^{(l+1)}-\psi^{(l)})^2$.
Therefore, by letting $B_l=\{x\in\cS:\,\psi_0(x)\ge \psi^{(l)}\}$,
we have
\begin{eqnarray}
\Dir(\psi_0) \ge \sum_{l=1}^{N}W(l)\, (\psi^{(l)}-\psi^{(l-1)})^2\, ,
\;\;\;\;\;\;\;
W(l)  \equiv  \sum_{(x,y)\in\partial B_l}\mu(x)p_{\beta}(x,y)\, .
\end{eqnarray}
On the other hand, $(\psi^{(i)})^2\le
i\,\sum_{l=1}^{i} (\psi^{(l)}-\psi^{(l-1)})^2$.
If $M(l) \equiv \sum_{x} \mu(x)
\ind(\psi_0(x) =\psi^{(l)}) = \mu(B_l)-\mu(B_{l-1})$
\begin{eqnarray}
\E(\psi_0^2) = \sum_{i=0}^NM(i)\;(\psi^{(i)})^2\le
\sum_{l=1}^{N}\Big(\sum_{i=l}^{N} i\, M(i)\Big)\;(\psi^{(l)}-\psi^{(l-1)})^2\, .
\end{eqnarray}
Therefore
\begin{eqnarray}
\lambda_{0,A} = \frac{\Dir(\psi_0)}{\E(\psi_0^2)}\ge
\inf_{1\le l\le N}\frac{W(l)}{\sum_{i=l}^{N} i\, M(i)}\, ,
\end{eqnarray}
which implies the thesis.
\end{proof}

\end{document}